\documentclass[letterpaper, 10 pt, conference]{ieeeconf}

\usepackage[utf8]{inputenc}    
\usepackage{amssymb}
\usepackage{amsfonts}
\usepackage{amsmath}
{}
\usepackage{graphicx}
\usepackage{mathtools}
\usepackage{xfrac}
\usepackage{nicefrac}
\usepackage{mathrsfs}
\usepackage{breqn}
\usepackage{dsfont}
\usepackage{wasysym}
\usepackage{tabularx}
\usepackage{relsize}
\usepackage{times}
\usepackage{url}
\usepackage{adjustbox}
\usepackage{boldline}
\usepackage{algorithm}
\usepackage[T1]{fontenc}
\usepackage[dvips]{epsfig}
\usepackage{epstopdf}
\usepackage{color}
\DeclareGraphicsExtensions{.ps,.eps,.pdf}

\newfloat{procedure}{htbp}{loa}
\floatname{procedure}{Procedure}

\IEEEoverridecommandlockouts

\newtheorem{remark}{Remark}
\newtheorem{definition}{Definition}
\newtheorem{lemma}{Lemma}

\renewcommand{\baselinestretch}{0.97}

\newcommand{\N}{\mathbb{N}}

\newcommand{\bx}{\boldsymbol{x}}
\newcommand{\be}{\begin{equation}}
\newcommand{\ee}{\end{equation}}
\newcommand{\ba}{\begin{array}}
\newcommand{\ea}{\end{array}}
\newcommand{\bitem}{\begin{itemize}}
\newcommand{\eitem}{\end{itemize}}
\newcommand{\benum}{\begin{enumerate}}
\newcommand{\eenum}{\end{enumerate}}
\onecolumn
\begin{document}
\date{}

\title{Hierarchical routing control in discrete manufacturing plants via\\ model predictive path allocation and greedy path following}

\author{Lorenzo~Fagiano, 
	Marko~Tanaskovic,
	Lenin~Cucas Mallitasig,
	Andrea~Cataldo and
	Riccardo~Scattolini%
	\thanks{L. Fagiano, L. Cucas and R. Scattolini are with the Dipartimento di Elettronica, Informazione e Bioingegneria, Politecnico di Milano, Piazza Leonardo da Vinci
32, 20133 Milano, Italy.}
\thanks{M. Tanaskovic is with Singidunum University, 32 Danijelova St., Belgrade, 160622 Serbia.}
\thanks{A. Cataldo is with the Institute of Industrial Technology and Automation,
	National Research Council, Via Alfonso Corti 12, 20133 Milano, Italy.}
\thanks{Corresponding author: L. Fagiano, lorenzo.fagiano@polimi.it.}
\thanks{This research was funded by a grant from the Italian Ministry of Foreign Affairs and International Cooperation (MAECI), project ``Real-time control and optimization for smart factories and advanced manufacturing''.}}

\maketitle

\begin{abstract}
The problem of real-time control and optimization of components' routing in discrete manufacturing plants, where distinct items must undergo a sequence of jobs, is considered. This problem features a large number of discrete control inputs and the presence of temporal-logic constraints. A new approach is proposed, adopting a shift of perspective with respect to previous contributions, from a Eulerian system model that tracks the state of plant nodes, to a Lagrangian model that tracks the state of each part being processed. The approach features a hierarchical structure. At a higher level, a predictive receding horizon strategy allocates a path across the plant to each part in order to minimize a chosen cost criterion. At a lower level, a path following logic computes the control inputs in order to follow the assigned path, while satisfying all constraints. The approach is tested here in simulations, reporting extremely good performance as measured by closed-loop cost function values and computational efficiency, also with very large prediction horizon values. These features pave the way to a number of subsequent research steps, which will culminate with the experimental testing on a pilot plant.
\end{abstract}

\section{Introduction}\label{s:intro}

Manufacturing is a key strategic sector in all industrialized countries. In many product categories, a high level of automation has enabled the mass production of goods with very high throughput and quality, and low unit cost. This is one of the main building blocks of modern economies. Yet, the strong global competition, together with the combined trends of higher product customization, more agile supply chains, and higher environmental sustainability, motivate further research and development in advanced manufacturing solutions \cite{eu_roadmap,SmartManufacturing2016,SmartManufacturing2018}. This interdisciplinary research domain involves many fields, from industrial communications to collaborative robotics, from human-machine interaction to routing and logistics, leading to a large number of interesting and challenging problems \cite{8438327,REN20191343,TUPTUK201893,YAN201756}. Among the latter, we focus on the real-time control and optimization of components' routing in discrete manufacturing plants, where distinct items must undergo a sequence of jobs. Depending on the specific manufacturing process at hand, this problem may entail several requirements. The discrete parts must be routed to a number of stations, via physical lines that present handling constraints, for example in terms of limited movement speed and potential line congestion. Different lines may also merge at some points, leading to possible lockouts to be avoided. Moreover, the processing time at each station may be uncertain within some limits, and the sequence of jobs to be done on a part may be not fully known a priori, since it can change depending on the outcome of each job. For example, some parts may need reworking due to a non-satisfactory outcome of one job, or they may undergo a random quality test at certain stages of the production process. In addition, the priorities among parts can also change in real-time, for example due to the segmentation of products and the need to increase flexibility and deliver production-on-demand. Other external factors such as component unavailabilities or faults may also affect the production lines. Finally, sustainability goals translate to minimization of waste and of energy consumption.
\begin{figure}[!htb]
	\centering
	\includegraphics[width=.8\columnwidth]{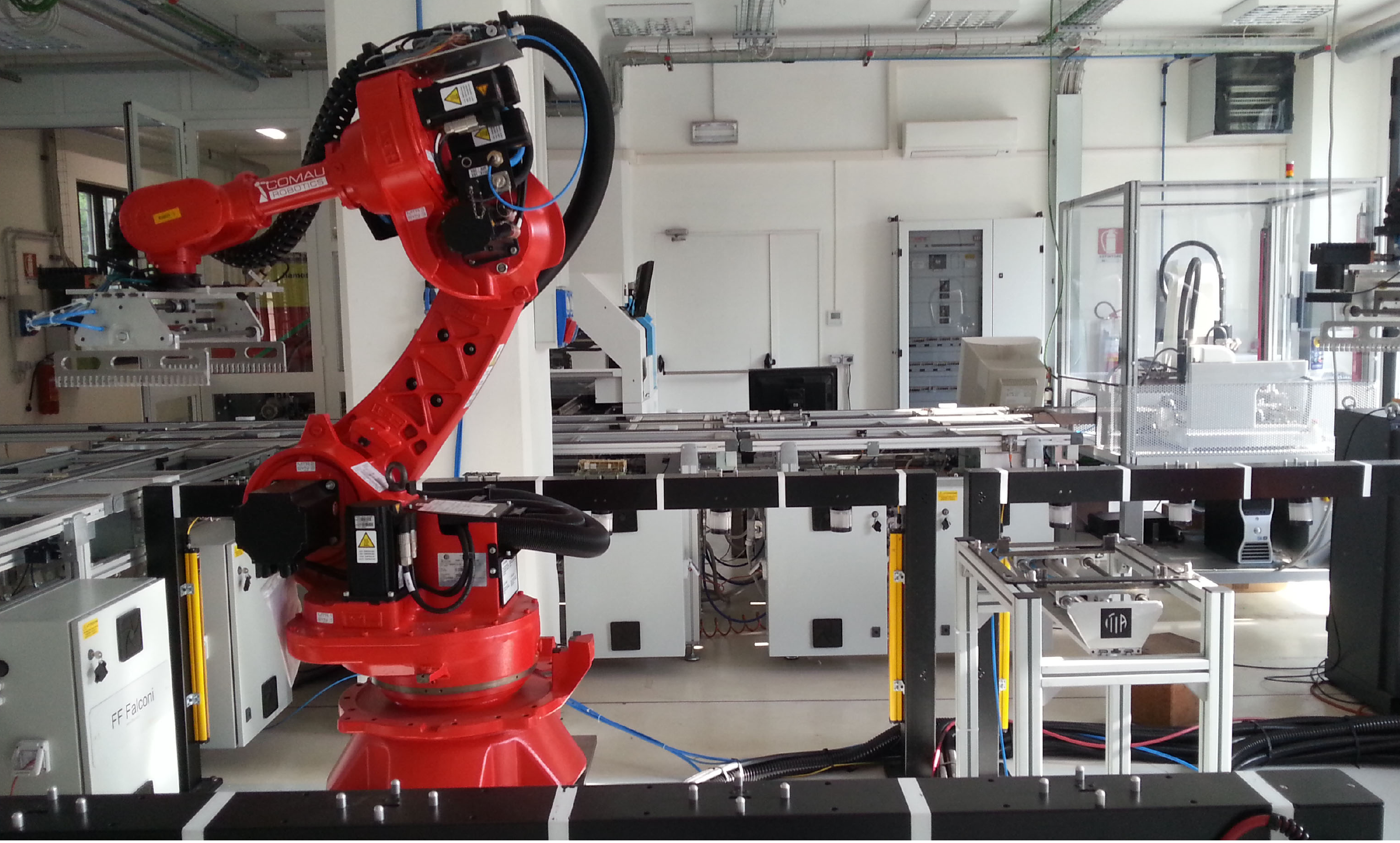}
	\caption{Laboratory de-manufacturing plant at the National Research Council in Milano (\cite{CaSc16}), showing the loading/unloading node with a manipulator, and the transportation nodes and a machine node in the background.}
	\label{F:pilot-plant}
\end{figure}

\noindent From a control engineering perspective, the mathematical transcription of this problem leads to a prohibitive large-scale integer or mixed-integer optimal control program, involving a dynamical system with discrete state variables, discrete input commands, and discrete output measurements, and subject to temporal logic constraints and external disturbance signals, where the goal is to guarantee the required throughput with minimal waste and energy cost. To tame such a complexity, hierarchical approaches are adopted: the overall problem is divided into sub-problems addressed separately, such as low level feedback control of individual machines and of segments of movement lines, computation of feasible routes, machine scheduling, etc.. Approaches in the literature that aim to address one or more of these sub-problems include rule-based techniques \cite{GuGuBe89,BYRNE1997109,SaChSi01,Bucki:jucs_21_4,Souier2010}, integer programming \cite{DAS1997237}, multi-agent architectures \cite{KoPiMe97}, short-term simulation and ordinal optimization \cite{PeCh98}, heuristic search combined with Petri nets \cite{MoYuKe02}, and model predictive control (MPC) \cite{VARGASVILLAMIL20002009,VARGASVILLAMIL200145,CATALDO201528,CaSc16,CaMS19}. In particular, in \cite{CaSc16} a receding horizon approach has been employed to control in real-time a de-manufacturing plant composed of 35 nodes (comprising either discrete movement elements or working/testing machines), accounting for temporal logic constraints and optimizing a multi-objective criterion that trades off the system throughput and the energy consumption. In an analogy with fluid modeling, this approach adopted a \textit{Eulerian description}, where the state vector includes the status of each node in the plant (which is conceptually similar to a control volume in fluid dynamics). The resulting control policy, experimentally tested on a laboratory setup (shown in Fig. \ref{F:pilot-plant}), has the merit of providing the optimal solution to the finite horizon routing problem at each time step. However, the drawback of this approach is the rapid increase of computational complexity with the number of prediction steps and of nodes in the plant, which limits its application to a relatively short prediction horizon and small system size.\\
This paper presents the first accomplished step of a research project aimed to improve over the results of \cite{CaSc16} in terms of scalability, while still satisfying the same, demanding temporal logic constraints. The main contribution presented here is a new approach to address the real-time routing problem,  with two novelties: 1) a shift of perspective from a  Eulerian  to a \textit{Lagrangian description}, where the system state includes the status of each part that must be routed in the plant, instead of each node; and 2) a hierarchical MPC structure, where the receding horizon strategy allocates a path to each part (as well as the part's position on the path) and a lower-level logic computes the control inputs in order to follow the assigned path. We tested the new approach in simulation and report extremely good performance as measured by closed-loop cost function values and computational efficiency, also with very large prediction horizon values. These features pave the way to a number of subsequent research steps, which will culminate with the experimental testing on the pilot plant of Fig. \ref{F:pilot-plant}. 


\section{Eulerian system model\\and problem description}\label{s:problem_statement}
We consider a discrete manufacturing plant composed of a finite number $N_n\in\N$ of nodes. At each discrete time instant $k$, each node $h=1,\ldots,N_n$ may be empty or it may host one (and only one) part being processed. For a reference, consider the diagram of Fig. \ref{F:example-scheme} representing the small-scale system that we use in this paper to test the proposed approach, composed of 12 nodes. A more complex diagram representing the laboratory testbed at the National Research Council in Milano, with 35 nodes, can be found in \cite{CaSc16}.\\
\begin{figure}[hbt!]
	\centering
	\includegraphics[width=.7\columnwidth]{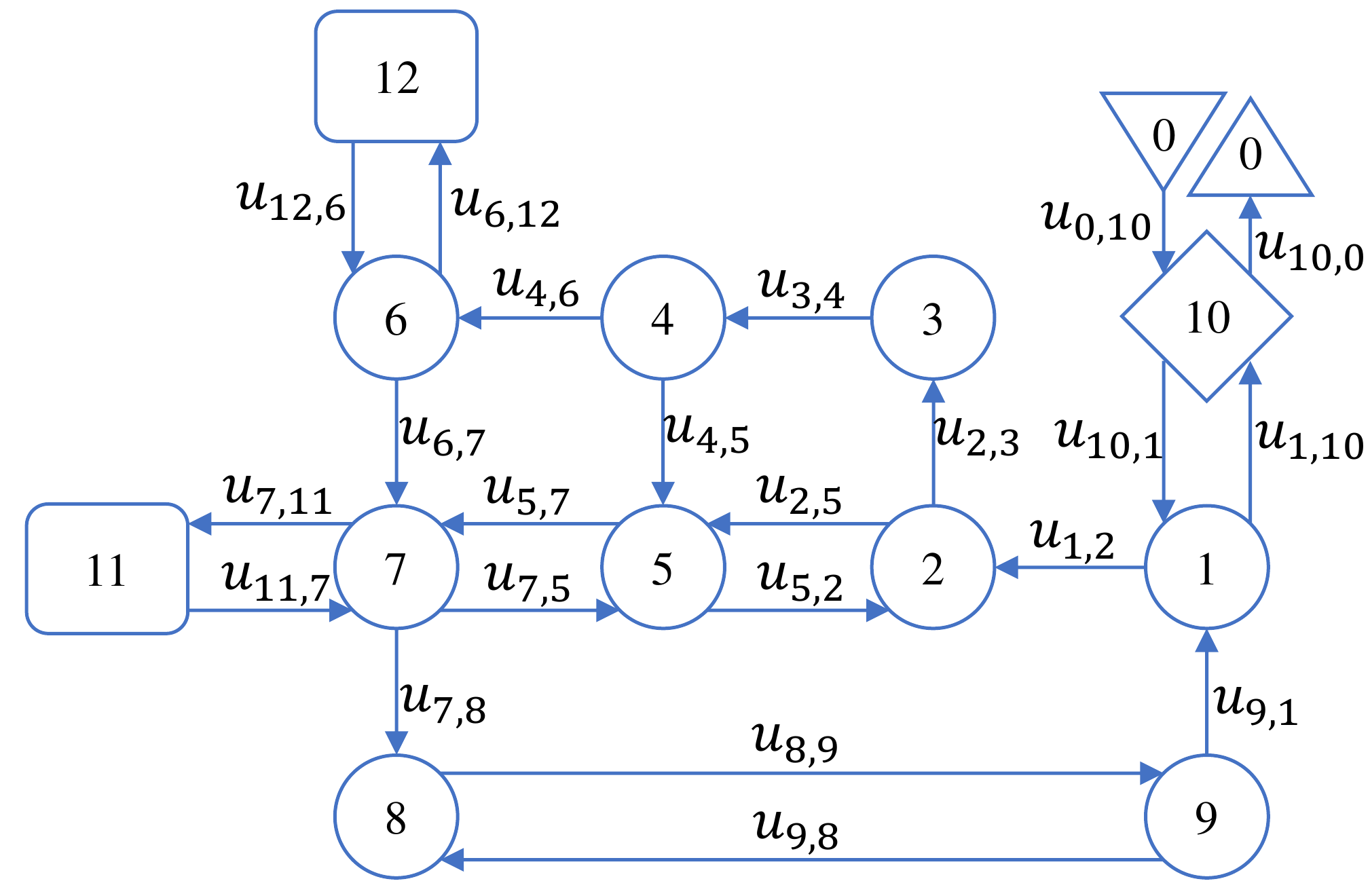}
	\caption{Small-scale system considered as test case in this paper. Node 10 is both the loading and unloading one (i.e., $h_l=h_u=10$) and nodes 11,12 are machines (i.e., $\mathcal{M}=\{11,12\}$).}\label{F:example-scheme}
\end{figure}
The boolean variable $z_h(k)\in\{0,1\}$ indicates whether a part is present at node $h$ (i.e., $z_h(k)=1$) or not. We assume that $N_t$ out of $N_n$ nodes are \textit{transportation modules}, and the remaining $N_m=N_n-N_t$ are \textit{machines}. In particular, let us denote the set of indexes of machine nodes as
\[
\mathcal{M}=\{h:\text{node $h$ is a machine}\}.
\]
Each node (be it a transportation module or a machine) is able to either hold one part in place, or to move it to a fixed number of specific, directly connected nodes according to the plant topology (see, e.g., Fig. \ref{F:example-scheme}).  Each machine must, in addition, execute a specific job on each part it receives. Without loss of generality we assume that a movement from one node to a connected one lasts one time step (direct movement), while the job carried out by a machine $m$ lasts an integer number $L_m\geq1$ of time steps. The boolean control signal $u_{h,j}(k)\in\{0,1\}$ dictates whether a part will move from node $h$ at time $k$ to node $j$ at time $k+1$. Finally, we also assume that two special nodes are present, the \textit{loading} node and the \textit{unloading} one, with indexes $h_l$ and $h_u$, respectively. These nodes are the interface between the plant under study and the outside, denoted with index 0, through two control variables:  $u_{0,h_l}$ can move to the loading node a part from outside the plant, e.g. from a buffer containing the incoming parts that must be processed, while a part can be moved from the unloading node to the outside via the control variable $u_{h_u,0}$, e.g. to a buffer of finished parts. We denote with $N_f(k)$ the total number of finished parts at time $k$. In summary, the number $N_u$ of boolean control signals to be computed at each time step is equal to the number of valid direct transitions among the nodes, plus the two loading and unloading commands $u_{0,h_l},\,u_{h_u,0}$. We collect these inputs into a column vector, denoted with $U(k)\in\{0,1\}^{N_u}$.

For each node $h=1,\ldots,N_n$, we define the following sets.
\begin{definition}\label{D:outgoing-incoming-sets} (Outgoing and Incoming sets)
	\bitem
	\item The \textit{outgoing set} $\mathcal{O}_h$ is the set containing the indexes of all nodes that can be reached directly from $h$, including possibly the outside, i.e. $\mathcal{O}_h=\{j:\,\exists\, u_{h,j}\}$;
	\item The \textit{incoming set} $\mathcal{I}_h$ is the set containing the indexes of all nodes for which $h$ is a direct destination, including possibly the outside, i.e. $\mathcal{I}_h=\{j:\,\exists\, u_{j,h}\}$.
	\eitem
\end{definition}
We thus have $\{0\}\in\mathcal{O}_{h_u}$ and $\{0\}\in\mathcal{I}_{h_l}$. Defining $\boldsymbol{z}=[z_1,\ldots,z_{N_n}]^T$ ($\cdot^T$ is the vector transpose operation) and
\[
\boldsymbol{v}(k)=\left[
\ba{c}
\sum\limits_{j\in\mathcal{I}_1}u_{j,1}(k)-\sum\limits_{j\in\mathcal{O}_1}u_{1,j}(k)\\
\vdots\\
\sum\limits_{j\in\mathcal{I}_{N_n}}u_{j,N_n}(k)-\sum\limits_{j\in\mathcal{O}_{N_n}}u_{N_n,j}(k)
\ea
\right]
\]
we can introduce the following linear model describing the plant's behavior:
\be\label{eq:model-plant}
\ba{rcrll}
\boldsymbol{z}(k+1)&=&\boldsymbol{z}(k)&+&\boldsymbol{v}(k)\\
N_f(k+1)&=&N_f(k)&+&u_{h_u,0}(k)
\ea
\ee
This model corresponds to a Eulerian description of the system, where the nodes are taken as control volumes, the system state corresponds to the number of parts in each of these volumes, and the model essentially corresponds to a series of mass conservation equations. To keep consistency with the real system, the boolean control inputs must comply with the following operational constraints at all time steps:
\begin{subequations}\label{eq:input-constr-always}
\begin{gather}
\sum\limits_{j\in\mathcal{O}_h}u_{h,j}(k)\leq1,\,h=1,\ldots,N_n\label{eq:input-constr-always-outgoing}\\
\sum\limits_{j\in\mathcal{I}_h}u_{h,j}(k)\leq1,\,h=1,\ldots,N_n\label{eq:input-constr-always-incoming}\\
\sum\limits_{j\in\mathcal{O}_h}u_{h,j}(k)=0,\,\forall h: z_h(k)=0\label{eq:input-constr-always-empty}\\
\sum\limits_{j\in\mathcal{I}_h}u_{j,h}(k)=0,\,\forall h: z_h(k)=1\land \sum\limits_{j\in\mathcal{O}_h}u_{h,j}(k)=0  \label{eq:input-constr-always-hold}
\end{gather}
\end{subequations}
Constraints \eqref{eq:input-constr-always-outgoing}-\eqref{eq:input-constr-always-incoming} impose that a part shall move to at most one destination from node $h$, and that only one part shall reach node $h$ at the next time step. Constraint \eqref{eq:input-constr-always-empty} states that all control signals from an empty node shall be zero, finally constraint \eqref{eq:input-constr-always-hold} imposes that no part can move to node $h$ if the latter is occupied and it will hold its current part in the next step.\\
\noindent Moreover, temporal logic constraints on the control inputs pertaining to machine nodes arise, due to the fact that once a job is started it must be completed before the part can be moved. Denoting with $k_m$ the time when a new job is started by machine $m$, such constraints take the form:
\be\label{eq:input-temporal-logic}
\sum\limits_{j\in\mathcal{O}_m}u_{m,j}(k)=0,\,\forall k\leq k_m+L_m,\forall m\in\mathcal{M}:z_m(k)=1
\ee
The problem we address can be described as follows: derive a control policy that computes, at each time instant $k$, all of the control variables $u_{h,j}$ in order to satisfy the operational constraints \eqref{eq:input-constr-always}-\eqref{eq:input-temporal-logic} and to minimize a suitably defined cost criterion.\\
In \cite{CaSc16}, this problem has been addressed resorting to MPC, after a suitable manipulation of the model and of the constraints that leads to a mixed logical dynamic (MLD) formulation and a large-scale mixed-integer linear program to be solved at each time step. The considered cost criterion was a weighted sum of terms that penalize the permanence of parts in the plant (thus encouraging a higher throughput) and the energy consumption as measured by a number of non-zero control inputs (which correspond to a physical movement of a part in the plant). The approach has been tested experimentally with good performance, however it suffers from the high computational complexity due to the large number of auxiliary integer and continuous variables that need to be introduced in the MLD reformulation. To give an example, in the small test case of Fig. \ref{F:example-scheme}, 160 integer auxiliary variables need to be introduced per each time step in the prediction horizon (e.g., with a 5-time-steps horizon about 800 integer variables are used).\\
The approach introduced in this paper, presented next, aims to overcome this issue by taking a different perspective on the problem at hand.

\section{Lagrangian system model\\and proposed approach}\label{s:approach}
To reduce the computational complexity while still retaining an optimization-based predictive approach, 
 we propose here the hierarchical control structure presented in Fig. \ref{F:hierarchical}:
\begin{figure}[hbt!]
	\centering
	\includegraphics[width=.4\columnwidth]{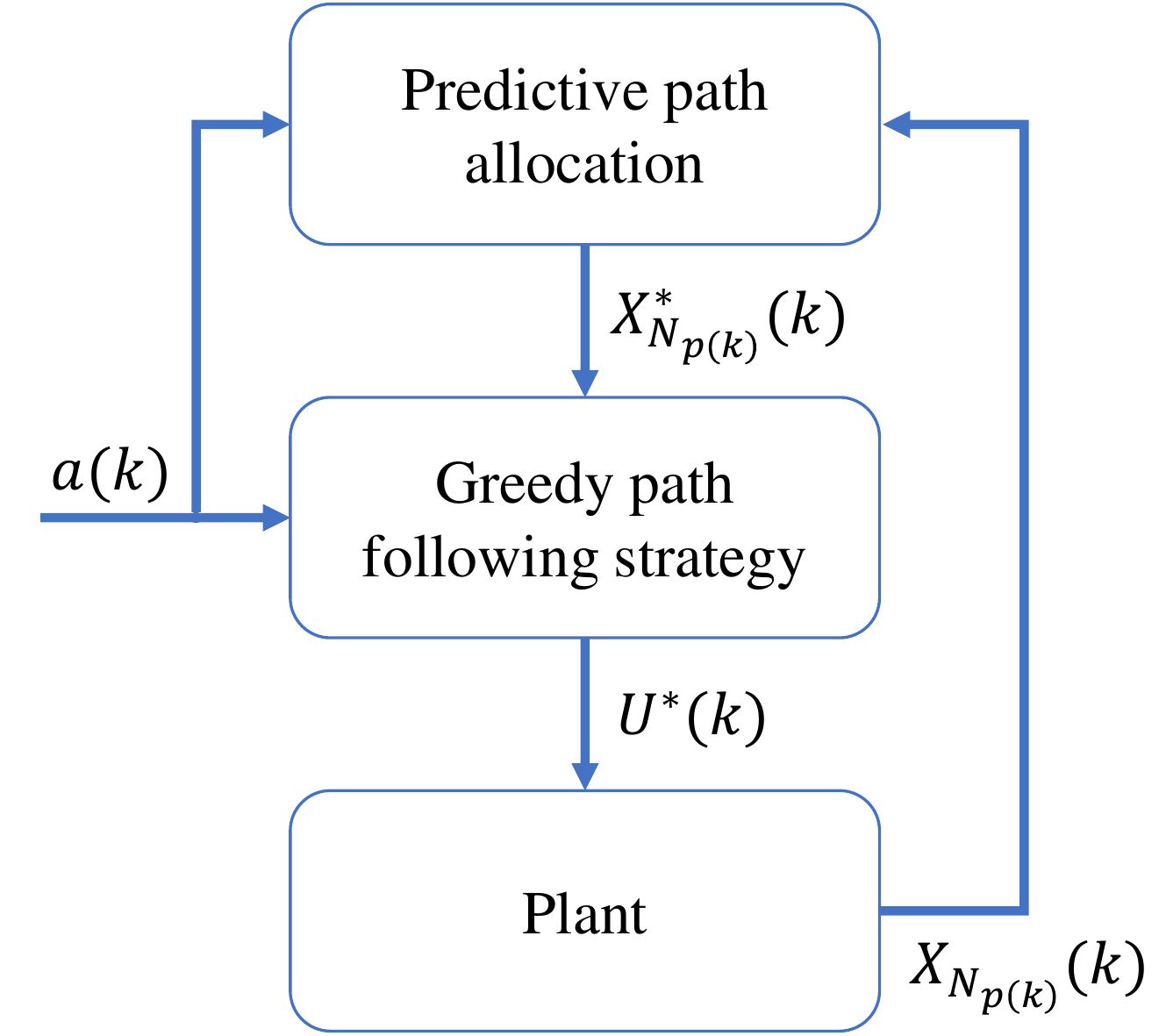}
	\caption{Hierarchical control approach proposed in this paper.}\label{F:hierarchical}
\end{figure}
\bitem
\item a \textit{low-level greedy path following strategy} is in charge to compute feasible control inputs and to move forward each part along its assigned path;
\item a \textit{high-level model predictive path allocation} module solves a finite horizon optimal control problem (FHOCP), where the decision variables are the sequences that each part shall follow, as well as their position in such sequences. Thus, the predictive controller can shift parts forward and backward on the various paths and move them from a  path to another one, as long as consistency with their physical location is maintained.
\eitem
As suggested by the adopted terminology, in this new approach we follow the parts' trajectories instead of keeping track of the status of each control volume, i.e. we adopt a Lagrangian description of the plant instead of an Eulerian one. Consistency/translation between the two descriptions is provided by the fact that each path is a sequence of nodes, thus the position of a part on a given path unequivocally identifies the node where that part is located. In the next sections, we present in detail the hierarchical approach by describing the two elements listed above.
\subsection{Lagrangian model state}
Let us denote with $i=1,\ldots,N_p(k)$ an index that identifies each one of the $N_p(k)\in\N$ parts in the plant at time $k$. The value of $N_p(k)\in\N$ can change over time as new parts enter the plant and/or finished ones exit. We further denote with $S=\{1,2,\ldots,N_s\}$ a set of integers, each one corresponding unequivocally to a \textit{sequence} (or path). Such sequences are assumed to be precomputed and stored: sequence generation and selection methods are not addressed in this paper, yet we will briefly comment on this aspect later on in Remark \ref{r:sequences}. For each $s\in S$, the operator $\mathcal{S}(s)$ returns the actual sequence corresponding to index $s$. Each sequence $\mathcal{S}(s)$ has the following structure:
\be\label{eq:sequence-structure}
\mathcal{S}(s)=\left\{\left[\ba{c}h_1\\g_1\ea\right],\ldots,\left[\ba{c}h_p\\g_p\ea\right],\ldots,\left[\ba{c}h_{N_s}\\g_{N_s}\ea\right]\right\}
\ee
where $N_s$ is the sequence length, $p=1,\ldots,N_s$ is the position along the sequence, and $h_p,\,g_p$ are integers corresponding to nodes in the plant. In particular, each value of $h_p$ corresponds to a node that is either equal to $h_{p+1}$ (i.e. the part shall be held), or directly connected to $h_{p+1}$ (i.e., the part shall be moved from node $h_p$ to $h_{p+1}$), while each value of $g_p$ is the index of a node chosen as goal for that part of the sequence. Usually, such goal indexes correspond to machines or to the outside (node 0).\\
We indicate with $s_i(k)\in S$ the sequence that part $i$ is following at time $k$, with $p_i(k)\in\N$ the position of part $i$ along such a sequence, and with $\mathcal{S}(s_i(k))^{(1,p_i(k))},\,\mathcal{S}(s_i(k))^{(2,p_i(k))}$ the first and second entry, respectively, of the vector in position $p_i(k)$ of sequence $\mathcal{S}(s_i(k))$ (compare \eqref{eq:sequence-structure}). For example, referring to Fig. \ref{F:example-scheme}, the sequence identified by index $s=1$ could correspond to:
\[
\ba{rcl}
\mathcal{S}(1)&=&\left\{\left[\ba{c}10\\12\ea\right],\left[\ba{c}1\\12\ea\right],\left[\ba{c}2\\12\ea\right],\left[\ba{c}3\\12\ea\right],\left[\ba{c}4\\12\ea\right],\right.\\&&\;\;\;
\left[\ba{c}6\\12\ea\right],\left[\ba{c}12\\12\ea\right],\left[\ba{c}12\\0\ea\right],\left[\ba{c}6\\0\ea\right],\left[\ba{c}7\\0\ea\right],\\&&\;\;\,
\left.\left[\ba{c}8\\0\ea\right],\left[\ba{c}9\\0\ea\right],\left[\ba{c}1\\0\ea\right],\left[\ba{c}10\\0\ea\right]\right\},
\ea
\]
and a part $i$ with $s_i(k)=1$ and $p_i(k)=3$ would be located at node $h=2$ at time $k$, i.e. $\mathcal{S}(s_i(k))^{(1,p_i(k))}=2$ and have as goal the machine node $\mathcal{S}(s_i(k))^{(2,p_i(k))}=12$. Moreover, we denote with $\underline{k}_i$ the time step when part $i$ appeared on the plant, and with $t_i(k)$ the time elapsed since then:
\be\label{eq:time-elapsed}
t_i(k)=k-\underline{k}_i.
\ee
Then, in our Lagrangian model the state of part $i$ is given by:
\be\label{eq:part-state}
\bx_i(k)=\left[\ba{c}s_i(k)\\p_i(k)\\t_i(k)\ea\right].
\ee
Finally, we denote with
\be\label{eq:priority}
r_i(k)=\text{card}(\mathcal{S}(s_i(k)))-p_i(k)
\ee
the number of remaining nodes that part $i$ shall visit to complete its current sequence. Variable $r_i(k)$ is thus a function of the state $\bx_i(k)$.\\
For later use, we also collect all the state variables in vector
\be\label{eq:lagr-state}
X_{N_p(k)}(k)=[\bx_1(k)^T,\ldots,\bx_{N_p}(k)^T]^T\in\N^{3N_p},
\ee
which represents the overall state of the Lagrangian plant model. Note that such a state vector can change dimension in time as it depends on the value of $N_p(k)$ (which we denote with the subscript $\cdot_{N_p(k)}$ in \eqref{eq:lagr-state}). Albeit rather unusual in dynamical models, this feature does not lead to any technical problem as long as consistency with the Eulerian model is ensured. In turn, this is obtained by always applying feasible inputs to the plant, as achieved by our path following algorithm, introduced next.
\subsection{Greedy path following strategy and closed-loop Lagrangian model}\label{ss:path-following}
The greedy path following strategy is a rule-based controller that acts according to the following principles: a) if possible, move each part forward in its current sequence; b) if the next node in the sequence is blocked, wait; c) if a potential conflict among parts is detected, the part with smallest $r_i(k)$ value shall move, and the other ones shall wait. To account for new parts that must be loaded to the plant from the outside, we introduce the boolean $a(k)$, which is equal to 1 when such a new part is available to be moved to the loading node.\\
$\,$\\
\textbf{Algorithm 1} \textit{Greedy path following strategy}. At each time step $k$:
\benum
\item[\textbf{1.}] Compute $r_i(k)$, $i=1,\ldots,N_p(k)$ according to \eqref{eq:priority};
\item[\textbf{2.}] Compute the one-step-ahead predicted states $\hat{\bx}_i,\,i=1,\ldots,N_p(k)$, by forward-propagation of all parts along their current paths:
\be\label{eq:forward-prop}
\ba{rcl}
\hat{p}_i(k+1)&=&p_i(k)+1\\
t_i(k+1)&=&t_i(k)+1\\
\hat{\bx}_i(k+1)&=&\left[\ba{c}s_i(k)\\\hat{p}_i(k+1)
\\t_i(k+1)\ea\right]
\ea
\ee
\item[\textbf{3.}] For each node $h=1,\ldots,N_n$, compute the number of potential conflicts $n_h(k+1)$ as:
\[
\ba{c}
n_h(k+1)=\sum\limits_{i=1}^{N_p(k)}c(\hat{\bx}_i(k+1),h)-1,\\\text{where}\\
c(\hat{\bx}_i(k+1),h)=
\left\{
\ba{l}1\text{ if }\mathcal{S}(1,s_i(k))^{(\hat{p}_i(k+1))}=h\\
0\text{ otherwise}
\ea
\right\};
\ea
\]
\item[\textbf{4.}] \verb|If| $n_h(k+1)=0$ for all $h$, \verb|then| go to step \textbf{5.}.\\\verb|Else|, for each node $h:n_h(k+1)>0$ do conflict resolution:
\benum
\item[\textbf{4.a.}] Compute the set containing the indexes of conflicting parts:
\[
\mathcal{C}_h(k)=\{i:c(\hat{\bx}_i(k+1),h)=1\}
\]
\item[\textbf{4.b.}] Check if a part is being held at node $h$:
\[
\ba{rcl}
\bar{i}_h(k)&=&\{i\in\mathcal{C}_h(k):\mathcal{S}(\hat{s}_{i}(k+1))^{(1,\hat{p}_{i^*}(k+1))}\\&&\;\;\;=\mathcal{S}(s_{i}(k))^{(1,p_{i}(k))}=h\}
\ea
\]
\item[\textbf{4.c.}] Compute the set of parts that are most advanced in their own path:
\[
\mathcal{L}_h(k)=\left\{i:r_i(k)=\min\limits_{l\in\mathcal{C}_h(k)}r_l(k)\right\}
\]
\item[\textbf{4.d.}] Compute the index $i_h^*(k)$ of the part with highest priority:\\
\verb|If| $\bar{i}_h(k)\neq\emptyset$ \verb|then| $i_h^*(k)=\bar{i}_h(k)$\\
\verb|Elseif| $\text{card}(\mathcal{L}_h(k))=1$ \verb|then| $i_h^*(k)=\mathcal{L}_h(k)$\\
\verb|Else| $i_h^*(k)=\arg\min\limits_{l\in\mathcal{L}_h(k)}t_l(k)$.\\
$\,$
\item[\textbf{4.e.}] $\forall i\in\mathcal{C}_h(k):i\neq i_h^*(k),$ correct the corresponding one-step-ahead predicted state as:
\be\label{eq:forward-prop-correct}
\ba{rcl}
\hat{p}_i(k+1)&=&p_i(k)\\
\hat{\bx}_i(k+1)&=&\left[\ba{c}s_i(k)\\\hat{p}_i(k+1)\\
t_i(k+1)\ea\right].
\ea
\ee
\item[\textbf{4.f.}] Go to \textbf{3.}
\eenum
\item[\textbf{5.}] Apply to the plant the following inputs, corresponding to the computed part movements: 
\[
\ba{l}
\forall h,j:\exists u_{h,j}\land h \neq 0\\
\ba{rcl}
u_{h,j}(k)&=&1\text{, if } \exists i:\mathcal{S}(s_{i}(k))^{(1,\hat{p}_{i}(k+1))}=j\\
&&\;\;\;\;\;\;\;\;\;\;\;\;\;\;\land\;\mathcal{S}(s_{i}(k))^{(1,p_{i}(k))}=h\\
u_{h,j}(k)&=&0\text{, else.}
\ea\\
\,\\
\ba{rcl}
u_{0,h_l}(k)&=&1\text{, if } a(k)=1\;\land\\
&&\;\;\;\;\;\;\;\nexists i:\mathcal{S}(s_{i}(k))^{(1,\hat{p}_{i}(k+1))}=h_l\\
u_{0,h_l}(k)&=&0\text{, else.}
\ea
\ea
\]
\eenum
$\,$\\
Step \textbf{4.d.} of \textbf{Algorithm 1} sets the priority as follows: a part being held at a node has the highest priority, if no part is held then the one that is most advanced in its own sequence (i.e. minimal $r_i(k)$) has the second-highest priority, if more than one part has minimal $r_i(k)$ then the one that has been in the plant for the longest time has the third-highest priority. Assuming that only one new part can enter the plant at each time (e.g., if there is only one loading node), this guarantees that eventually only one part is selected and advanced. If more than one loading node exist, then another condition (e.g. based on part number) can be easily implemented to sort out possible ambiguities.
\begin{lemma}\label{lemma:feasibility} (Recursive feasibility of \textbf{Algorithm 1}).
	Assume that at a given time $\underline{k}$ at most one part is present at each node, and that for any index $s\in S$, if the corresponding sequence $\mathcal{S}(s)$ includes a machine node $m\in\mathcal{M}$ at some position $p$, and a node $h\neq m$ at position $p-1$, then such a machine node appears at least $L_m$ times consecutively, i.e. inside $\mathcal{S}(s)$ there is a subsequence
	\[
	\left[\ba{c}m\\g_p\ea\right],\left[\ba{c}m\\g_{p+1}\ea\right],\ldots,\left[\ba{c}m\\g_{p+v}\ea\right]
	\]
	with $v\geq L_m$.
	Then, the inputs computed by \textbf{Algorithm 1} satisfy the constraints \eqref{eq:input-constr-always}-\eqref{eq:input-temporal-logic} for all $k\geq\underline{k}$.
\end{lemma}
\begin{proof}
A sketch of the proof is provided for the sake of compactness. At time $\underline{k}+1$, constraints \eqref{eq:input-constr-always} are satisfied by the conflict-resolution logic of \textbf{Algorithm 1}, which always ends with at most only one incoming part at each node, except for nodes where a part is being held. Constraint \eqref{eq:input-temporal-logic} is satisfied by the assumption that each sequence containing a machine node $m$ features that node repeated consecutively for at least $L_m$ positions, which results in a part being held at least $L_m$ time steps in machine $m$. Feasibility of the inputs at all time steps $k>\underline{k}+1$ is obtained by induction.
\end{proof}
The feedback control policy defined by \textbf{Algorithm 1} corresponds to a set of functions $\kappa_{N_p},\,N_p\in\N$. Each one of these functions pertains to a specific number of parts $N_p$ and its input arguments are the corresponding Lagrangian state $X_{N_p(k)}$ and signal $a(k)$, while the output of all of them is a vector of plant commands $U\in\{0,1\}^{N_u}$ (see step \textbf{5.} of \textbf{Algorithm 1}):
\be\label{eq:fb-greedy}
U(k)=\kappa_{N_p(k)}(X_{N_p(k)}(k),a(k)).
\ee
As shown in Lemma \ref{lemma:feasibility}, such a control policy generates inputs that are always feasible under a rather mild assumption, since the sequences $\mathcal{S}(s),\,s\in S,$ are selected/computed by the designer, who can easily enforce the property required by Lemma \ref{lemma:feasibility}. However, input feasibility by itself does not prevent the controlled system from running into a lockout, and in general the greedy path following approach can give suboptimal behavior with respect to the performance criteria of interest. On the other hand, when combined with the model \eqref{eq:model-plant}, \textbf{Algorithm 1} allows one to predict the system behavior without having to explicitly enforce the challenging constraints \eqref{eq:input-constr-always}-\eqref{eq:input-temporal-logic} and at extremely low computational cost. We exploit such closed-loop predictions in a high-level MPC strategy, described in the next Section.\\
Before proceeding further, we also introduce the 
closed loop Lagrangian model of the system, provided by the following algorithm.\\
$\,$\\
\textbf{Algorithm 2} \textit{Closed-loop Lagrangian plant model}. At each time step $k$:
\benum
\item[\textbf{1.}] Run \textbf{Algorithm 1} with the current values of $X_{N_p(k)}(k)$ and $a(k)$ as inputs, collect all the resulting values of $\hat{\bx}_i(k+1)$ and $u_{h,j}(k),\,\forall(h,j):\exists u_{h,j}$;
\item[\textbf{2.}] Compute the Lagrangian state dimension $N_p(k+1)$ as
\[
N_p(k+1)=N_p(k)+u_{0,h_l}-u_{h_u,0};
\] 
\item[\textbf{3.}] If $u_{0,h_l}=1$, generate the state $\tilde{\bx}(k+1)$ of the new part that will be loaded to the plant at time $t+1$;
\item[\textbf{4.}] For all $i:\mathcal{S}(s_{i}(k))^{(1,\hat{p}_{i}(k+1))}\neq h_l$, compute the state $\bx_i(k+1)=\hat{\bx}_i(k+1)$.
\item[\textbf{5.}] Compute the Lagrangian state $X_{N_p(k+1)}(k+1)$ by stacking all vectors $\bx_i(k+1)$ computed at step \textbf{4.} and, if available, vector $\tilde{\bx}(k+1)$ computed at step \textbf{3.}. Set $k=k+1$ and go to \textbf{1.}.
\eenum
$\,$\\
The state initialization of a new part at step \textbf{3.} can be done by assigning a sequence $s\in S$ and position $p$ to it (typically, but not necessarily, $p=1$), and by setting the third state equal to zero (compare \eqref{eq:time-elapsed}-\eqref{eq:part-state}). Since their state value is not updated at step \textbf{4.}, parts that are unloaded from the plant naturally disappear from the Lagrangian model.\\
Similarly to the control policy \eqref{eq:fb-greedy}, the system model defined by \textbf{Algorithm 2} also corresponds to a set of functions, $f_{(N_p^+,N_p)}:3\N^{N_p}\rightarrow 3\N^{N_p^+}$, each one pertaining to a specific pair of part quantities, i.e. those at the current and at next time steps, while the signal $a(k)$ is an exogenous input:
\be\label{eq:closed-loop-lagrangian-model}
X_{N_p(k+1)}(k+1)=f_{(N_p(k+1),N_p(k))}(X_{N_p(k)}(k),a(k)).
\ee
Equation \eqref{eq:closed-loop-lagrangian-model} highlights the fact that the Lagrangian model describes the motion of the parts, whose number can increase or decrease from one step to the next depending on the number of newly loaded parts and of unloaded ones.

\subsection{Model predictive path allocation}\label{ss:path-allocation}
At each time step, the predictive control logic chooses whether to keep each part on its current path $s_i(k)$ and at its current position $p_i(k)$, or to change one or both of these elements in order to optimize the predicted plant performance. The result is a dynamic, optimization-based path allocation strategy that can exploit very large prediction horizon values, thus guaranteeing the absence of lockouts, and allows one to easily consider different performance indexes and to generally improve the plant behavior with respect to the one obtained by the greedy path following policy alone.\\
At each time step $k$, let us consider the following sets $\mathcal{X}_i(k),\,i=1,\ldots,N_p(k)$:
\begin{subequations}
\begin{align}\label{eq:opt-domain-nodes}
&\text{if } \mathcal{S}(s_i(k))^{(1,p_i(k))}\notin\mathcal{M}:\nonumber\\
&\mathcal{X}_i(k)=\left\{\ba{l}(s,p)\in S\times\N:\\\mathcal{S}(s)^{(1,p)}=\mathcal{S}(s_i(k))^{(1,p_i(k))}\\
\land\,\mathcal{S}(s)^{(2,p)}=\mathcal{S}(s_i(k))^{(2,p_i(k))}\ea\right\}\\
\,\nonumber\\
&\text{else if } \mathcal{S}(s_i(k))^{(1,p_i(k))}\in\mathcal{M}:\nonumber\\
&\mathcal{X}_i(k)=\left\{\ba{l}(s,p)\in S\times\N:\\\mathcal{S}(s)^{(1,p-j)}=\mathcal{S}(s_i(k))^{(1,p_i(k)-j)},\\
j=0,\ldots,k-k_{\mathcal{S}(s_i(k))^{(1,p_i(k))}}\\
\land\,\mathcal{S}(s)^{(2,p)}=\mathcal{S}(s_i(k))^{(2,p_i(k))}\ea\right\}
\end{align}
\end{subequations}
where $k_{\mathcal{S}(s_i(k))^{(1,p_i(k))}}$ is the time step when part $i$ started the job of machine $m=\mathcal{S}(s_i(k))^{(1,p_i(k))}$ (compare \eqref{eq:input-temporal-logic}). Namely, each set $\mathcal{X}_i(k)$ contains all the pairs $(s,p)$ of sequence index and position index such that the  corresponding vector $\left[\mathcal{S}(s)^{(1,p)}\\\mathcal{S}(s)^{(2,p)}\right]^T$ corresponds to that of part $i$ at time $k$, also considering a possible ongoing job and its remaining duration, if $\mathcal{S}(s_i(k))^{(1,p_i(k))}$ is a machine node. These sets are never empty by construction, since they always include the current pair $(s_i(k),\,p_i(k))$. At any time $k$, exchanging these two components of the state $\bx_i(k)$ to any other pair $(s,p)\in\mathcal{X}_i(k)$ implies that we are allocating to part $i$ another sequence and/or position among those that are compatible with its current physical location and goal. For example, in this way it is possible to select one out of several parallel paths originating from a certain node in the plant, or to make a part wait or repeat several times a single loop in the plant, each time by shifting it back in a sequence that contains that loop only once. Our high-level predictive controller exploits precisely this feature, as described in the following algorithm. We denote with $X_{N_p(o|k)}(o|k),\,\bx(o|k)$ the predictions of plant input and Lagrangian states, respectively, computed at time $k$ and pertaining to time $k+o$.\\
$\,$\\
\textbf{Algorithm 3} \textit{Model Predictive Path Allocation}. 
\benum
\item[\textbf{1.}] At time $k$ acquire the state variables $\bx_i(k),\,i=1,\ldots,N_p(k)$ and compute the corresponding sets $\mathcal{X}_i(k)$;
\item[\textbf{2.}] Solve the following finite horizon optimal control problem (FHOCP):
\begin{subequations}\label{eq:FHOCP}
	\begin{gather}
	\min\limits_{(\sigma_i,\pi_i),\,i=1,\ldots,N_p(k)}
	\sum\limits_{o=0}^{N}\ell_{N_p(o|k)}\left(X_{N_p(o|k)}(o|k)\right)\label{eq:FHOCP-cost}\\
	\text{subject to}\nonumber\\
	\bx_i(0|k)=\left[\sigma_i,\pi_i,t_i(k)\right]^T,\,i=1,\ldots,N_p(k)\label{eq:FHOCP-each-state}\\
	\,\nonumber\\
	\ba{c}X_{N_p(0|k)}(0|k)=\\\left[\bx_1(0|k)^T,\ldots,\bx_{N_p(k)}(0|k)^T\right]^T\ea\label{eq:FHOCP-init-state}\\
	\,\nonumber\\
	\ba{c}X_{N_p(o+1|k)}(o+1|k)=\\
	f_{(N_p(o+1|k),N_p(o|k))}(X_{N_p(o|k)}(k),a(o|k)),\\
	o=0,\ldots,N-1
	\ea\label{eq:FHOCP-prediction}\\
	\,\nonumber\\
	(\sigma_i,\pi_i)\in\mathcal{X}_i(k),\,i=1,\ldots,N_p(k)\label{eq:FHOCP-constraint}
	\end{gather}
\end{subequations}
where $N\in\N$ is the prediction horizon, the sequence $a(o|k)\in\{0,1\},\,o=0,\ldots,N-1$ contains the predictions of new parts that need to be worked (if available), and the stage cost functions $\ell_{N_p}\left(X_{N_p}\right)$ are chosen by the designer according to the plant performance indicator of interest.
\item[\textbf{3.}] Let $(\sigma_i^*,\pi_i^*),\,i=1,\ldots,N_p(k)$ be the solution to \eqref{eq:FHOCP}. Compute the new state vectors $\bx^*_i(k)$ as:
\[
\ba{rcl}
\bx_i^*(k)&=&\left[\ba{c}\sigma_i^*\\\pi_i^*\\t_i(k)\ea\right],i=1,\ldots,N_p(k)\\
X^*_{N_p(k)}(k)&=&\left[\bx_1^*(k)^T,\ldots,\bx^*_{N_p(k)}(k)^T\right]^T
\ea
\]
and provide these values to \textbf{Algorithm 1} to compute the control inputs via \eqref{eq:fb-greedy}:
\[
U^*(k)=\kappa_{N_p(k)}(X^*_{N_p(k)}(k),a(k)).
\]
\item[\textbf{4.}] Apply to the plant the control inputs $U^*(k)$, set $k=k+1$, go to \textbf{1.}.
\eenum
The predictive control strategy defined by \textbf{Algorithm 3} is thus able to directly modify the state of each part that is fed to the path following strategy (see Fig. \ref{F:hierarchical}), ensuring consistency with its current positions and goal (constraint \eqref{eq:FHOCP-constraint}), in order to optimize the chosen performance index \eqref{eq:FHOCP-cost} on the basis of a prediction of the plant behavior under the greedy path following algorithm, see \eqref{eq:FHOCP-each-state}-\eqref{eq:FHOCP-prediction}. Note that the optimization variables $(\sigma_i,\pi_i),i=1,\ldots,N_p(k)$ pertain only to the current time step, i.e. the sequence index is not changed during the predictions. This clearly reduces the degrees of freedom of the solver, resulting in possible sub-optimality but gaining in computational efficiency, similarly to what is done in move blocking strategies in MPC, see e.g. \cite{CAGIENARD2007563}. On the other hand, being a receding horizon strategy, \textbf{Algorithm 3} is able to change the sequence and position indexes $s_i(k),p_i(k)$ of all states at each time step $k$, resulting in practice in good closed-loop performance. Signal $a(k)$, which is managed by the greedy path following algorithm as described in Section \ref{ss:path-following}, is considered as an external disturbance, of which a prediction may be available (otherwise one can simply set $a(o|k)=0$ in \eqref{eq:FHOCP-prediction}).\\
Regarding the choice of cost functions $\ell_{N_p}\left(X_{N_p}\right)$, possible examples include the sum, over all parts, of the remaining steps in their respective sequences (which favors plant throughput), plus the sum of non-zero control inputs (which favors energy saving). The design of cost functions accounting for different real-world requirements is subject of current research.\\ As regards the guaranteed closed-loop performance, the optimization problem \eqref{eq:FHOCP} is always feasible by construction, since to a minimum the controller can just leave sequence and position indexes unchanged, and plant constraints are always satisfied by the path following approach. On the other hand, a sensible question pertains to what we refer to as the \textit{lockout avoidance} property, i.e. the guarantee that the predictive approach always prevents occurrence of a lockout. Under mild assumptions on the chosen sequences and prediction horizon $N$,we can indeed prove that \textbf{Algorithm 3} guarantees lockout avoidance. This result and its proof are omitted here for the sake of brevity.
\begin{remark}\label{r:sequences} (Computation of node sequences) The performance of the closed-loop plant under the proposed hierarchical approach strongly depend on the pre-computed paths. The generation of these paths entails a trade-off between two conflicting aspects: on the one hand, a large number of comprehensive paths will provide the predictive controller with more degrees of freedom to accommodate more parts and reach higher performance, on the other hand a set of sequences that is too rich can lead to very high computational complexity, reducing the scalability of the proposed approach. In the numerical example presented in this paper, where each part has to visit the two machines one after the other, we adopted a manual selection based on physical insight, see Section \ref{s:example}. We plan to rigorously investigate the problem of optimal sequence computation and selection in the next future, adopting approaches from graph theory combined with closed-loop system analysis. 
\end{remark}

\section{Numerical results}\label{s:example}
We present the tests of the hierarchical approach on the small-scale example of Fig. \ref{F:example-scheme}, with $N_u=22$ boolean control inputs. The two machine nodes $11,\,12$ have the same processing time $L_{11}=L_{12}=3$ time steps, and each part must visit first machine 12, then machine 11 before leaving the plant from node 10, which is also the loading node. We assume that $a(k)=1\,\forall k$, i.e. a new part is loaded to the plant whenever the loading node is free, and that the predictive controller does not have this information. The maximum throughput of the plant depends on the processing time of machine 12 and on the fact that node 10 has to switch between loading a new part or unloading a finished one, thus adding two additional time steps. Its value is thus equal to $1/(L_{12}+2)=0.20$ parts per time step. We set a prediction horizon of $N=50$ time steps, and we use as stage cost in \eqref{eq:FHOCP-cost} the following function:
\be\label{eq:example_stage_cost}
\ell_{N_p(o|k)}=\sum\limits_{i=1}^{N_p(o|k)}r_i(o|k)+\beta\sum\limits_{i=1}^{N_p(o|k)}\sum\limits_{i=1}^{n_u}U(o|k)
\ee
where $\beta\geq0$ is a weighting factor, $r_i(o|k)$ is computed as in \eqref{eq:priority} considering the predicted Lagrangian states, and $U(o|k)=\kappa_{N_p(o|k)}(X^*_{N_p(o|k)}(o|k),0)$ is the vector of simulated actuation commands given to the plant. Function \eqref{eq:example_stage_cost} is thus the weighted sum of two objectives: the total number of remaining steps in the sequence assigned to each part, related to throughput maximization, and the total number of commanded inputs, related to energy minimization. As regards the sequence computation, since the considered example is essentially a series manufacturing process, we adopt here a single path, composed of redundant sub-sequences going several times through all possible loops across nodes $2,3,4,5,6,7$ (see Fig. \ref{F:example-scheme}) and of sub-sequences of identical values for each node, in order to provide the predictive controller with the option to make one part wait in place by shifting it back with such sub-sequences.
The criteria that we considered in the path generation are: the inclusion in each sequence of all the machines in the correct order, the inclusion of subsequences as required by Lemma \ref{lemma:feasibility}, and the inclusion in each sequence of a terminal sub-path leading to the outside of the plant (unloading node).\\
\begin{figure}[!htb]
	\centering
	\includegraphics[width=.8\columnwidth]{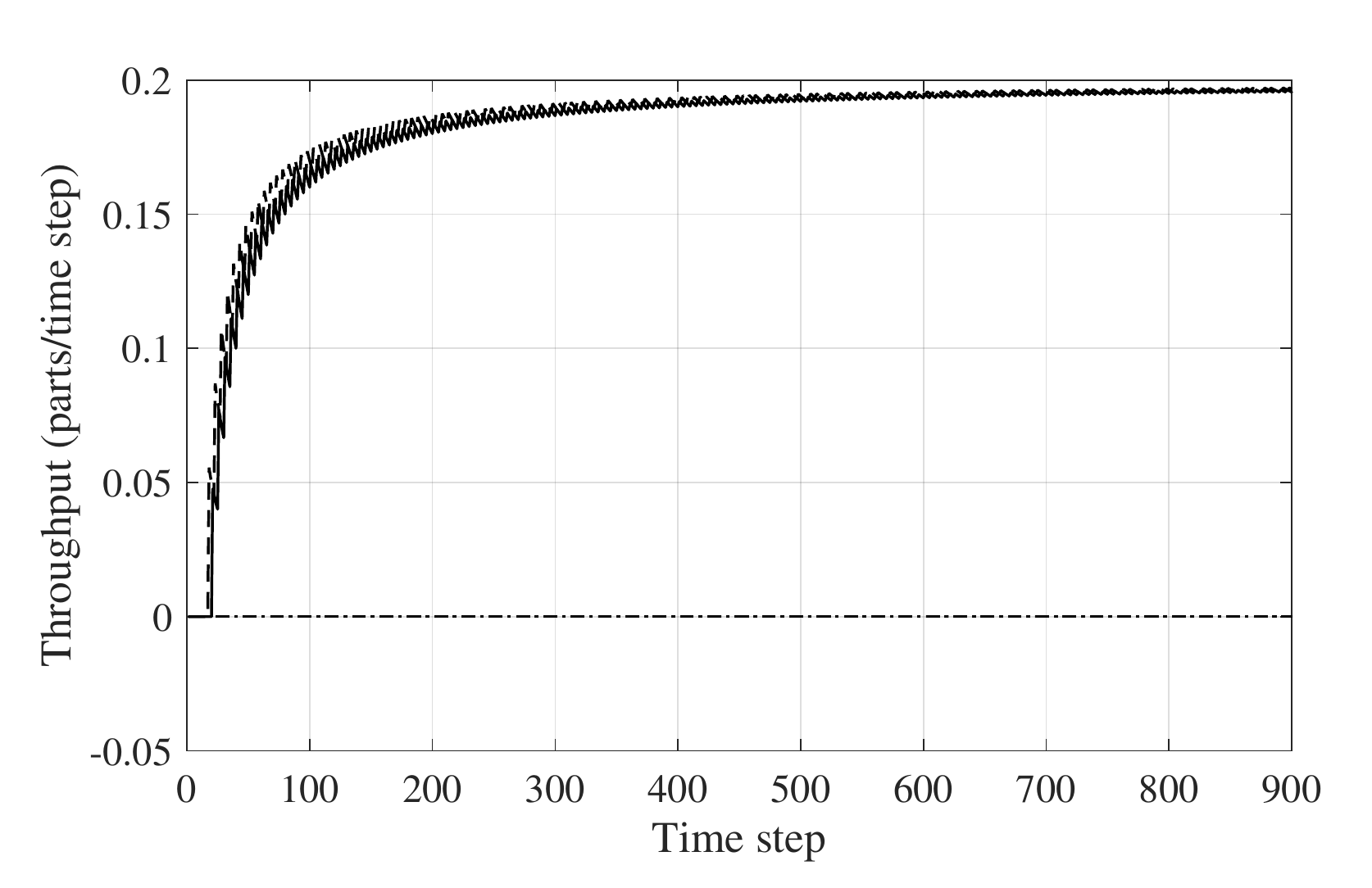}
	\caption{Simulation example. Course of the plant throughput expressed as number of finished parts per time step, with $\beta=5$ (dashed line), $\beta=6$ (solid), and $\beta=8$ (dash-dotted).}
	\label{F:throughput}
\end{figure}
\begin{figure}[!htb]
	\centering
	\includegraphics[width=.8\columnwidth]{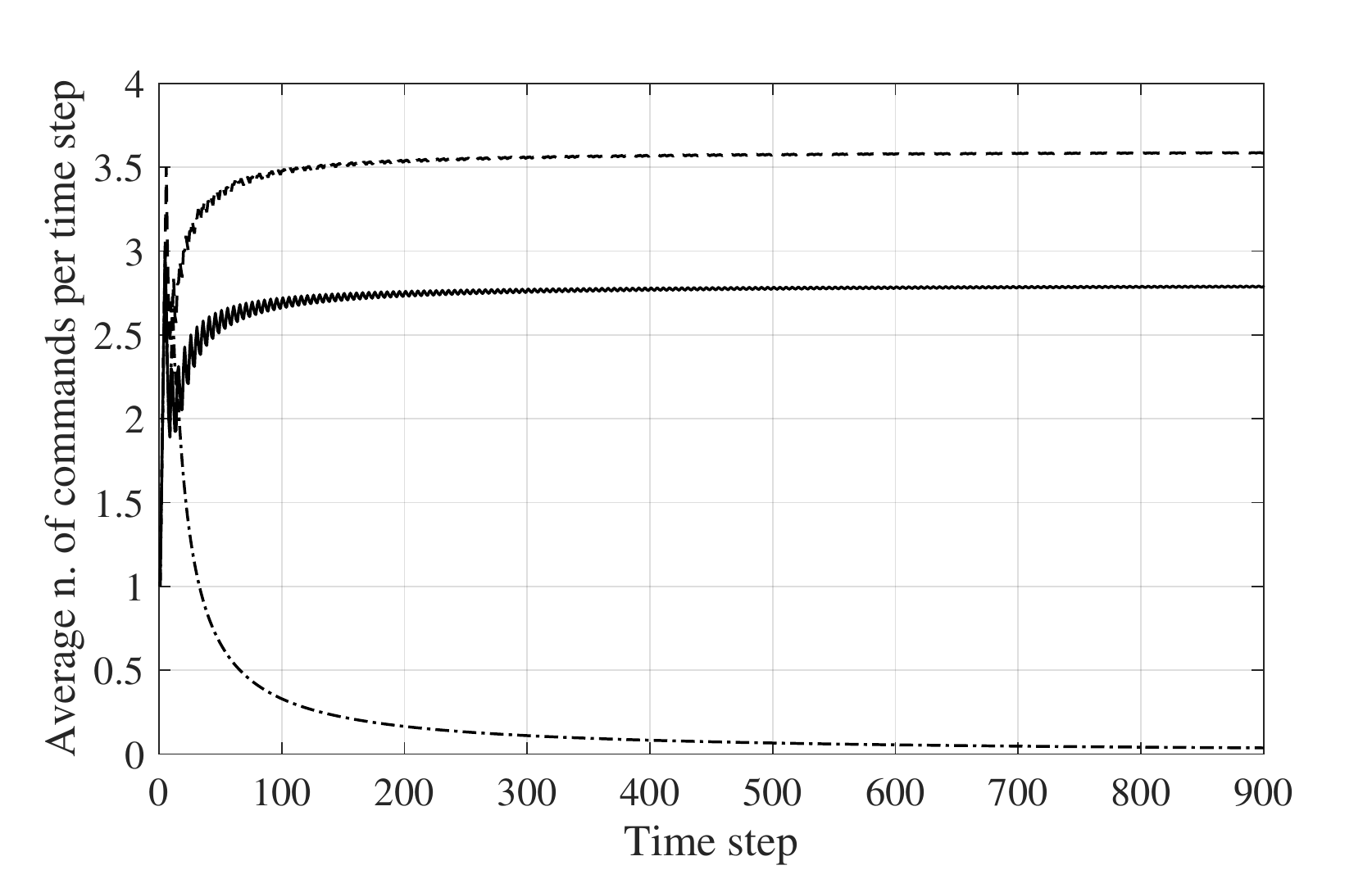}
	\caption{Simulation example. Course of the average number of input commands per time step, with $\beta=5$ (dashed line), $\beta=6$ (solid), and $\beta=8$ (dash-dotted).}
	\label{F:commands}
\end{figure}
We ran all the simulations starting from one part in node 10. In this example, the maximum throughput can be reached with different strategies that lead to different values of energy consumption: in fact, during each job it is possible to let the waiting parts be held on the nodes, or to make them circulate in the available loops within the plant. We illustrate that the proposed strategy switches between these two behaviors as the value of $\beta$ decreases. This is clearly visible in Figs. \ref{F:throughput}-\ref{F:commands}: with $\beta=5$ the plant reaches the maximum throughput and a number of commands per time step equal to 3.5, while with $\beta=6$ the same throughput is obtained with only 2.5 commands per time step, i.e. 30\% less. In both cases, a number of parts oscillating between 7 and 8 is present on the plant at each time step, after the initial transient. If we further increase $\beta$, the controller reaches a lockout with eight parts on the plant, since it becomes more convenient to avoid any actuation rather than to push the parts forward in their paths. This is also shown in Figs. \ref{F:throughput}-\ref{F:commands}.\\
Finally, regarding the computational aspects, we solved the problem \eqref{eq:FHOCP} via extensive search over all possible valid $(\sigma_i,\pi_i)$ pairs. On a Laptop with 8GB RAM and an Intel Core i7 CPU at 2.6 GHz running Matlab, the resulting computational time is 0.45$\,$s per time step, without any attempt to improve the solver efficiency (e.g. by parallelizing the computations and/or adopting a non-brute-force approach to solve the optimization problem). 
\section{Conclusions}\label{s:conclusions}
A new approach to the problem of routing parts in discrete manufacturing plants has been presented, adopting a Lagrangian modeling perspective and a hierarchical control structure. Simulation results on a small example illustrate the behavior of the closed loop system, which achieves the theoretical maximum throughput and allow one to optimize energy consumption, as measured by the number of actuated commands. The obtained computational times are very low for the considered application, also considering the rather large employed prediction horizon. This makes us confident about the scalability to larger plants. Next steps in this research are aimed to investigate the generation of optimal sequences, the derivation of theoretical guarantees about lockout avoidance, testing in scenarios with uncertain outcomes of each job and non-series manufacturing processes, and the experimental validation on a pilot plant.
\bibliographystyle{IEEEtran}

\end{document}